\documentclass[sn-mathphys-num]{sn-jnl}

\usepackage{amsmath,amsfonts}
\usepackage{amssymb,amsthm}
\usepackage{makecell}
\usepackage{geometry}

\theoremstyle{thmstylethree}
\newtheorem{theorem}{Theorem}
\newtheorem{example}{Example}
\newtheorem{remark}{Remark}
\newtheorem{definition}[theorem]{Definition}
\newtheorem{lemma}[theorem]{Lemma}
\newtheorem{corollary}[theorem]{Corollary}

\begin{document}
	
\title{Extension of Optimal Locally Repairable codes}
\author[1]{\fnm{Yunlong} \sur{Zhu}}\email{zhuylong3@mail2.sysu.edu.cn}
\author*[1,2]{\fnm{Chang-An} \sur{Zhao}}\email{zhaochan3@mail.sysu.edu.cn}

\affil*[1]{\orgdiv{School of Mathematics}, \orgname{Sun Yat-sen University}, \orgaddress{\city{Guangzhou}, \postcode{510275}, \state{Guangdong}, \country{China}}}

\affil[2]{\orgname{Guangdong Key Laboratory of Information Security}, \orgaddress{\city{Guangzhou}, \postcode{510006}, \state{Guangdong}, \country{China}}}

\abstract{
Recent studies have delved into the construction of locally repairable codes (LRCs) with optimal minimum distance from function fields.  In this paper, we present several novel constructions by extending the findings of optimally designed locally repairable codes documented in the literature. Let $C$ denote an optimal LRC of locality $r$, implying that every repairable block of $C$ is a $[r+1, r]$ MDS code, and $C$ maximizes its minimum distance. By extending a single coordinate of one of these blocks, we demonstrate that the resulting code remains an optimally designed locally repairable code.  This  suggests that the maximal length of an optimal LRC from rational function fields can be extended up to $q+2$ over a finite field $\mathbb{F}_q$. In addition, we give a new construction of optimal $(r, 3)$-LRC by extending one coordinate in each block within $C$. Furthermore, we propose a novel family of LRCs with Roth-Lempel type that are optimal under certain conditions. Finally,  we explore optimal LRCs derived from elliptic function fields and extend a single coordinate of such codes. This approach leads us to confirm that the new codes are also optimal, thereby allowing their lengths to reach $q + 2\sqrt{q} - 2r - 2$ with locality $r$.  We also consider the construction of optimal $(r, 3)$-LRC in elliptic function fields, with exploring one more condition.}

\maketitle
\section{Introduction}
Among various coding techniques employed in distributed and cloud storage systems, locally repairable codes (LRCs for short, which was introduced by Gopalan et al. in \cite{Gopalan}) have emerged as a prominent solution due to their ability to efficiently repair lost or corrupted nodes with minimal resource utilization. Suppose that $\mathbb{F}_q$ is a finite field with $q$ elements. Then we say that an $[n,k,d]$ linear code $C$ over $\mathbb{F}_q$ is an LRC with locality $r$ if for a codeword ${\bf c}=(c_1,\ldots,c_n)$, the value of any coordinate $c_i$ can be recovered by at most $r$ other coordinates of the codeword {\bf c}. The concept of LRCs has two generalizations. One concerns the local recovery of several coordinates(\cite{Galindo,Gao,Xing}), the other concerns codes with multiple recovery sets (\cite{Bartoli,Chara-mul,Haymaker,Jin-mul,Munuera}).

Let $\delta\ge2$ be a positive integer. An $[n,k,d]$ linear code $C$ is called $(r,\delta)$-LRC if for any $i\in\{1,\ldots,n\}$, there exists a subset $I_i$ such that $|I_i|\le r+\delta-1$ and any $\delta-1$ values of this block can be recovered by other $|I_i|-\delta+1$ values in $I_i$. The famous Singleton bound for classical codes asserts that
\[
	d\le n-k+1
\]
for an $[n,k,d]$ linear code $C$. Suppose that $C$ has locality $r$. Then it was proved in \cite{Gopalan} that
\[
	d\le n-k-\lceil\frac{k}{r}\rceil+2.
\]
The bound of the minimum distance is called the Singleton-type bound of LRCs. For $(r,\delta)$-LRCs, the Singleton-type bound was proved in \cite{Kamath} as
\[
	d\le n-k-(\lceil\frac{k}{r}\rceil-1)(\delta-1)+1.
\]
A locally repairable code $C$ is called $d$-optimal if the minimum distance of $C$ achieves the upper Singleton-type bound. In this paper, we call that $C$ is optimal for short.

A general construction of optimal LRCs from rational function fields is given by Tamo et al. in \cite{Tamo-family}. They proposed the concept of "good polynomial" which has degree $r+1$ and is constant on each block in the partition of the selected points set. Suppose that $C$ is an LRC constructed by this construction, and any repairable block of $C$ can be seen as an MDS code. The construction of LRCs via good polynomials have attracted the attention of many researchers (\cite{Ruikai-good,Ruikai-function,Dukes,Gao,Liu}). Barg et al. lifted this idea to function fields with higher genus and gave the general construction from algebraic curves in \cite{Barg-lrcalg}. Barg et al. also gave a general construction of optimal LRCs from function field in \cite{Barg-alg}. Jin et al. \cite{Jin-optimal} have given an explicit construction of optimal LRCs via the rational function field. They also used the modified algebraic codes and the length of the codes can attain $q+1$ over $\mathbb{F}_q$. Li et al. have used the automorphism groups of elliptic curves to construct optimal LRCs and the length of codes can attain $q+2\sqrt{q}$ when $q$ is even. Note that the locality of their constructions is upper bounded by 23. Ma and Xing \cite{Ma-optimal} then investigated group structures of automorphism groups of elliptic curves and gave LRCs with locality greater than 23. Kim \cite{Kim} studied the LRCs from Hermitian function fields with some divisors and Chara et al. \cite{Chara} gave the construction of LRCs from towers of function fields.

{\it Our Contributions:} On the basis of optimal LRCs from rational function fields or elliptic function fields, we  propose  some novel families of LRCs by extending a single coordinate. Suppose that $C$ of  length $n$ is an optimal LRC given in \cite{Barg-alg,Jin-optimal,Li-optimal,Ma-optimal} from function fields. Note that each repairable block in such a code $C$ can be regarded as an MDS code with length $r+1$. By extending one of these repairable blocks in the code $C$ by adding a single coordinate, we obtain some novel LRCs of length $n+1$. More precisely,  by using the modified algebraic codes given in \cite{Jin-optimal}, we give an LRC with length $q+2$ over $\mathbb{F}_q$ and prove the optimality of this code. Another interesting result presented in this paper involves the construction of $(r,3)$-LRCs. By extending the previously known optimal LRCs for $s$ coordinates—where $s$ corresponds to the number of repairable blocks, we demonstrate that they are $(r,3)$-optimal. Furthermore, we explore the extension of Roth-Lempel type LRCs from optimal LRCs in rational function fields and show that under certain conditions, these LRCs also remain optimal. 
In the case of optimal LRCs from elliptic function fields, we consider a new function space different from \cite{Li-optimal} for a single coordinate extension and define a new code. The new function space can be regarded as a subspace of the Riemann-Roch space for a divisor in elliptic function field. The new function space allows us to prove the locality and optimality of the new code. In the case of $(r,3)$-locality with $s$ coordinates extension, by 
noting that divisors of degree $zero$ in elliptic function fields are not always principal divisors, we need one more condition for the optimality.

\begin{table}[htbp]
	\renewcommand{\arraystretch}{1.5}
	\centering
	\caption{The parameters of some known LRCs from function fields}
	\begin{tabular}{ccccccc}
		\hline
		Ref. & Field size &Length $n$ & Dimension $k$ & Minimum Distance $d$ & Locality $r$ & Optimal\\
		\hline
		\cite[Proposition 4.1]{Barg-alg} & $q^2$ & $q^3$ & $(\ell+1)(q-1)$ & $\ge n-\ell q-(q-2)(q+1)$ & $q-1$ & not optimal\\
		\hline
		\cite[Proposition 4.2]{Barg-alg} & $q^2$ & $q^3-q$ & $(\ell+1)q$ & $\ge n-\ell(q+1)-q(q-1)$ & $q$ & not optimal\\
		\hline
		\cite[Section 3]{Barg-lrcalg} & $q$ & $\#E(\mathbb{F}_q)-(r+1)$ & $tr$ & $n-(tr+t+r)$ & arbitrary & not optimal\\
		\hline
		\cite[Section 3]{Barg-lrcalg} & $q$ & $\#E(\mathbb{F}_q)$ & $tr$ & $n-(2r'+2t'(r+1))$ & arbitrary & not optimal\\
		\hline
		\cite[Section 4]{Jin-optimal} & $q$ & $s(r+1)$ & $rt$ & $n-rt-t+2$ & \makecell{$r=p^v-1$\\$(r+1)|(q-1)$\\$r+1=up^v$} & optimal\\
		\hline
		\cite[Section 5]{Jin-optimal} & $q$ & $s(r+1)$ & $rt$ & $n-rt-t+2$ & $(r+1)|(q+1)$ & optimal\\
		\hline
		\cite{Li-optimal} & $q$ & $s(r+1)$ & $rt-r+1$ & $n-(t-1)(r+1)$ & $r=2,3,5,7,11,23$ & optimal\\
		\hline
		\cite{Li-optimal} & $q$ & $3s$ & $2t+1$ & $n-3t$ & $r=2$ & optimal\\
		\hline
		This paper & $q$ & $s(r+1)+1$ & $rt$ & $n-rt-t+2$ & same to \cite{Jin-optimal} & optimal\\
		\hline
		This paper & $q$ & $s(r+1)+1$ & $rt-r+1$ & $n-rt-t+2$ & same to \cite{Li-optimal} & optimal\\
		\hline
		This paper & $q$ & $s(r+2)$ & $rt$ & $n-t(r+1)-s+3$ & same to \cite{Jin-optimal} & $(r,3)$-optimal\\
		\hline
		This paper & $q$ & $s(r+2)$ & $rt-r+1$ & $n-(t-1)(r+1)-s+1$ & same to \cite{Li-optimal} & $(r,3)$-optimal\\
		\hline
	\end{tabular}
\end{table}

The organization of this paper is given as follows. In Section II we recall some necessary concepts and notions, and we also recall the constructions of LRCs given in \cite{Barg-alg} (or in \cite{Jin-optimal}) and \cite{Li-optimal}. In Section III, we give our results about LRCs from rational function fields. In Section IV, we consider the extension of optimal LRCs by Roth-Lempel type. Then in Section V, we present our construction of optimal LRCs from elliptic function fields.

\section{Preliminaries}
In this section, we recall some basic definitions and notions on locally repairable codes, function fields, algebraic geometry codes, RS-type codes, etc. Throughout this paper, we assume that $q$ is a prime power and $\mathbb{F}_q$ is a finite field. 
\subsection{Locally Repairable Codes}
The code $C\subseteq\mathbb{F}_q^n$ is said an LRC code with locality $r$ if the value of every coordinate for a given codeword can be recovered by accessing at most $r$ other coordinates of this codeword. The formal definition of LRC is given as follows.
\begin{definition}
Let $C\subseteq\mathbb{F}_q^n$ be a linear code with length $n$ and dimension $k$. For a given $\alpha\in\mathbb{F}_q$, consider the sets of codewords
\[
C(i,\alpha):=\{{\bf c}=(c_1,\ldots,c_n)\in C|c_i=\alpha\}.
\]
For a subset $I\subseteq\{1,\ldots,n\}$, we denote by $C_I(i,\alpha)$ the restriction of $C(i,\alpha)$ on $I$. The code $C$ is said to be {\it locally recoverable with locality r} if for each $i\in\{1,\ldots,n\}$, there exists a subset $I_i\subseteq\{1,\ldots,n\}$ with $|I_i|\le r$ such that $C_{I_i}(i,\alpha)$ and $C_{I_i}(i,\beta)$ are disjoint for any $\beta\neq\alpha$. The subset $I_i$ is called the {\it recover set} of the coordinate $i$.
\end{definition}
Let $\delta\ge2$ be a positive integer, then the concept of LRCs can be extended to $(r,\delta)$-LRCs as follows:
\begin{definition}
Let $C\subseteq\mathbb{F}_q^n$ be a linear code with length $n$ and dimension $k$. If for any $i\in\{1,\ldots,n\}$, there exists a subset $I_i\subseteq\{1,\ldots,n\}$ with $|I_i|\le r+\delta-1$ such that the restriction code $C_{I_i}$ has minimum distance at least $\delta$, then $C$ is called a $(r,\delta)$-LRC with locality $(r,\delta)$.
\end{definition}
Note that if $\delta=2$, then the above two definitions are equivalent. When $\delta\ge3$, the definition of $(r,\delta)$ means that we can recover $\delta-1$ values from other $\le r$ values of any code in $C_{I_i}$.

For applications of LRCs, the codes with small locality, high rate of $\frac{k}{n}$ and large minimum distance $d$ are more preferred. The locality $r$ is usually upper bounded by $k$, and if we allow $r=k$, the Singleton-type bound of minimum distance becomes the classical Singleton bound.


\subsection{Algebraic geometry code}
Let $X$ be an absolutely irreducible, smooth algebraic curve defined over $\mathbb{F}_q$, with genus $g(X)$. Let $F(X)$ denote the function field associated with $X$. For a rational divisor $G$ on $X$, the Riemann-Roch space associated with $G$ is
\[
\mathcal{L}(G) := \{f\in F(X)\backslash\{0\}:\text{\rm div}(f)+G\ge 0\} \cup \{0\},
\]
and the dimension of $\mathcal{L}(G)$ is denoted by $\ell(G)$. Let $P_1,\ldots,P_n$ be $n$ pairwise distinct rational points of $X$, where $P_i\notin \text{\rm supp}(G)$ for all $i$. For any $f\in F(X)$, assume that div$(f)=\sum\limits_Pn_{P}P$, where the valuation $v_{P}(f)=n_{P}$. In particular, if $f\in\mathcal{L}(G)$, we have $v_{P_i}(f)\ge0$ for all $i$ with $1\le i\le n$.

Let $D=P_1+\cdots+P_n$ be an effective divisor, and consider the evaluation map:
\begin{align*}
	ev_D:  \mathcal{L}(G)&\to \mathbb{F}_q^n,\\
	f&\mapsto (f(P_1),\ldots,f(P_n)). 
\end{align*}
The AG code denoted by $C_L(D,G)$, represents the image of $ev_D$. Then the parameters of $C_L(D,G)$ are:
\[
k=\ell(G)-\ell(G-D),\ d\ge n-\text{\rm deg}(G),
\]
where $d^*=n-\text{\rm deg}(G)$ is called the {\it design distance} of $C_L(D,G)$. It is straightforward to verify that $ev_D$ is an embedding when $\text{\rm deg}(G)<n$ and $k=\ell(G)$. By the Riemann-Roch theorem, if $\text{\rm deg}(G)\ge 2g(X)-1$, then we obtain:
\[
\ell(G)= \text{\rm deg}(G)-g(C)+1,
\]
and the minimum distance
\[
d\ge n-k-g(C)+1.
\]

Jin et al. removed the condition that $P_i\notin {\rm supp}(G)$ for all $i$ in \cite{Jin-optimal} as follows. Suppose that $m_i=v_{P_i}(G)$. Choose a prime elements $\pi_i$ of $P_i$ for $1\le i\le n$. Then for any function $f\in\mathcal{L}(G)$, we have
\[
	v_{P_i}(\pi_i^{m_i}f)=m_i+v_{P_i}(f)\ge m_i-v_{P_i}(G)=0.
\]
Define the modified AG code as follows
\[
	C_L^m(D,G):=\{((\pi_1^{m_1}f)(P_1),\ldots,(\pi_n^{m_n}f)(P_n))|f\in\mathcal{L}(G)\}.
\]
It was shown that $C_L^m(D,G)$ is still an $[n,k,d]$ linear code in \cite{Jin-optimal}.

\subsection{Function field extension}
Let $F/\mathbb{F}_q$ be an algebraic function field over $\mathbb{F}_q$. Assume that $Aut(F/\mathbb{F}_q)$ is the automorphism group of $F$ over $\mathbb{F}_q$, and $\mathcal{G}\subseteq Aut(F/\mathbb{F}_q)$. Define the fixed subfield of $F/\mathbb{F}_q$ with respect to $\mathcal{G}$ as
\[
	F^{\mathcal{G}}/\mathbb{F}_q:=\{y\in F/\mathbb{F}_q|\forall\sigma\in\mathcal{G}, \sigma(y)=y\}.
\]
Then $F/F^{\mathcal{G}}$ is a Galois extension with $Gal(F/F^{\mathcal{G}})=\mathcal{G}$, and $[F:F^{\mathcal{G}}] = |\mathcal{G}|$. Denote by $g(F^{\mathcal{G}})$ the genus of $F^{\mathcal{G}}$ over $\mathbb{F}_q$. Consider the finite separable extension $F/\mathbb{F}_q$ of function field $F^{\mathcal{G}}/\mathbb{F}_q$, the Hurwitz Genus Theorem \cite[Th. 3.4.13]{Stich} yields
\[
	2g(F)-2=|\mathcal{G}|(2g(F^{\mathcal{G}})-2)+{\rm deg(Diff}(F/F^{\mathcal{G}}))
\]
where
\[
	{\rm Diff}(F/F^{\mathcal{G}}):=\sum\limits_{P\in\mathbb{P}_F}\sum\limits_{P'|P}d(P'|P)\cdot P'
\]
is the different of $F/F^{\mathcal{G}}$. We recall the following lemma (\cite{Barg-alg} and \cite{Li-optimal}):
\begin{lemma}\label{funcfield}
	Let $F/\mathbb{F}_q$ be a functional field with the full constant field $\mathbb{F}_q$. Suppose that $F'/\mathbb{F}_q$ is a subfield of $E$ such that $F/F'$ is a finite separable extension.
	\begin{itemize}
		\item If $F$ is a rational function field, then $F'$ is again a rational function field.
		\item If $F$ is an elliptic function field. Moreover, there is a place $Q$ of $F$ with ramification index $e_Q\ge 2$. Then $F'$ is a rational function field.
	\end{itemize}
\end{lemma}
Before presenting our constructions, let us recall some general constructions of LRCs in the following two subsections, which are given in \cite{Barg-alg} and \cite{Li-optimal} respectively.

\subsection{A general construction of optimal LRCs}\label{section2.4}
In the subsection, we will recall a general family of optimal LRC codes (\cite{Barg-alg} and \cite{Jin-optimal}) as follows. Let $E/\mathbb{F}_q$ and $F/\mathbb{F}_q$ be rational function fields such that $E/F$ is a finite separable extension with $[E:F]=r+1$. Assume that $Q_1,\ldots,Q_s$ are $s$ rational places of $F/\mathbb{F}_q$, and they are splitting completely in $E/F$. For each $Q_i$, let $P_{i,1},P_{i,2},\ldots,P_{i,r+1}$ be the $r+1$ rational places of $E/\mathbb{F}_q$ lying over $Q_i$. Denote the places at infinity of $F/\mathbb{F}_q$ and $E/\mathbb{F}_q$ by $Q_{\infty}$ and $P_{\infty}$ respectively. Choose an integer $t$ with $1\le t\le s$, and let $f_1,\ldots,f_t$ be a basis for $\mathcal{L}((t-1)Q_{\infty})$. For the place $P_{\infty}$, we can find a function $z\in E/\mathbb{F}_q$ such that $(z)_{\infty}=P_{\infty}$ by \cite[Prop.1.6.6]{Stich}. Then we can define a set of functions
\[
V := \{\sum\limits_{i=0}^{r-1}\sum\limits_{j=1}^ta_{i,j}f_jz^i | {\bf a}=(a_{i,j})\in\mathbb{F}_q^{rt}\}.
\]
It is clear that $V$ is a linear function space over $\mathbb{F}_q$ with dimension $k=rt$. Moreover, there exists a divisor $G$ such that $V\subseteq\mathcal{L}(G)$, and
\[
{\rm deg}(G)=(t-1)(r+1)+r-1=rt+t-2.
\]
From the definition of the AG codes, we are able to present codes with the function space $V$ as follows:
\[
	C := \{(f_{\bf a}(P_{1,1}),\ldots,f_{\bf a}(P_{s,r+1}))| f_{\bf a}\in V\}
\]
where $f_{\bf a}$ is the function with respect to ${\bf a}=(a_{i,j})\in\mathbb{F}_q^{rt}$. By \cite{Barg-alg}, we have that the code $C$ is an optimal $[s(r+1),rt]$ LRC with locality $r$. Motivated by the concept of the extended RS codes in \cite[Chap.5]{Huffman-fecc}, we will present a new construction of optimal LRCs by extending a certain block from optimal LRCs, and a new construction of $(r,\delta)$-LRCs by extending each block from optimal LRCs.

\subsection{A construction of optimal elliptic LRCs}\label{section2.5}
Suppose that $\mathcal{E}$ is an elliptic curve defined over $\mathbb{F}_q$, and $E/\mathbb{F}_q$ is the function field defined by $\mathcal{E}$. Let $\mathcal{G}$ be a subgroup of $Aut(E/\mathbb{F}_q)$ with $\#\mathcal{G}=r+1$, and $F=E^{\mathcal{G}}$ be the fixed subfield of $E/\mathbb{F}_q$ by $\mathcal{G}$. Then the field extension $E/F$ is a Galois extension with $[E:F]=r+1$. Suppose that $Q_0,Q_1,\ldots,Q_s$ are $s+1$ places of $F/\mathbb{F}_q$ and they are all splitting completely in $E/F$. We have the following lemma.
\begin{lemma}\cite[Proposition 20]{Li-optimal}\label{eclrc}
Let $P_{u,1},\ldots,P_{u,r+1}$ be the $r+1$ places of $E/\mathbb{F}_q$ lying on $Q_u$ for each $0\le u\le s$. Then
\begin{itemize}
\item There exists a function $z$ in $E$ such that $F/\mathbb{F}_q=\mathbb{F}_q(z)$ and $(z)_{\infty}=P_{0,1}+\cdots+P_{0,r+1}$.
\item There exists $r-1$ elements $\omega_i\in E/\mathbb{F}_q$ with $(\omega_i)_{\infty}=P_{0,1}+\cdots+P_{0,i+1}$ for $1\le i\le r-1$, such that $\omega_0=1,\omega_1,\ldots,\omega_{r-1}$ are linearly independent over $F/\mathbb{F}_q$.
\item Let $M_u$ be an $r+1\times r$ matrix over $\mathbb{F}_q$ with
\begin{equation}\nonumber
M_u=\begin{pmatrix}
1&\omega_1(P_{u,1})&\omega_2(P_{u,1})&\ldots&\omega_{r-1}(P_{u,1})\\
1&\omega_1(P_{u,2})&\omega_2(P_{u,2})&\ldots&\omega_{r-1}(P_{u,2})\\
\vdots&\vdots &\vdots &\ddots &\vdots\\
1&\omega_1(P_{u,r+1})&\omega_2(P_{u,r+1})&\ldots&\omega_{r-1}(P_{u,r+1})
\end{pmatrix}.
\end{equation}
Then every $r\times r$ sub-matrix of $M_u$ is invertible for all $1\le u\le s$.
\end{itemize}
\end{lemma}
Consider the function space
\[
V_{E} := \{\sum\limits_{j=1}^{t}a_{0,j}z^{j-1}+\sum\limits_{i=1}^{r-1}(\sum\limits_{j=1}^{t-1}a_{i,j}z^{j-1})\omega_i | a_{i,j}\in\mathbb{F}_q\}
\]
with dimension $rt-r+1$. Let $f_{\bf a}$ is the function with respect to ${\bf a}=(a_{i,j})\in\mathbb{F}_q^{rt-r+1}$. Then the code
\[
C := \{(f_{\bf a}(P_{1,1}),\ldots,f_{\bf a}(P_{s,r+1}))| f_{\bf a}\in V_E\}
\]
is an optimal LRC with length $n=s(r+1)$, dimension $k=rt-r+1$ and locality $r$ by \cite{Li-optimal}.
\section{Extension of optimal LRCs on rational function fields}
In this section we present the construction of extend RS-type optimal LRC code on rational function fields. Let $E/\mathbb{F}_q$ is a rational functional field and $\mathcal{G}\subset Aut(E/\mathbb{F}_q)$ is an automorphism subgroup with order $r+1$. Suppose that $F/\mathbb{F}_q=E^{\mathcal{G}}$ be the subfield of $E/\mathbb{F}_q$. Then by Lemma \ref{funcfield} we obtain that $F/\mathbb{F}_q$ is also a rational function field, and $E/F$ is a finite separable extension with $[E:F]=r+1$. Let $Q_1,\ldots,Q_s,Q_{\infty}$ and $P_{1,1},\ldots,P_{s,r+1},P_{\infty}$ are defined as Section \ref{section2.4}.
\subsection{Construction of extended-RS type LRCs}\label{section3.1}
Suppose that $t$ is an integer with $t\ge2$. Since $\ell((t-1)Q_{\infty}-Q_1)=t-1$, we can choose $f_1,\ldots,f_{t-1}$ such that $f_j\in\mathcal{L}((t-1)Q_{\infty}-Q_1)$ are linearly independent over $F/\mathbb{F}_q$. The set $\mathcal{L}((t-1)Q_{\infty})$ has size $q^t$ and the set $\mathcal{L}((t-1)Q_{\infty}-Q_1)$ has size $q^{t-1}$, therefore there exists a function $f_t\in\mathcal{L}((t-1)Q_{\infty})$ such that $f_1,\ldots,f_{t}$ are linearly independent. Let $z\in E/\mathbb{F}_q$ with $(z)_{\infty}=P_{\infty}$. We also define the function space
\[
V := \{\sum\limits_{i=0}^{r-1}\sum\limits_{j=1}^ta_{i,j}f_jz^i | (a_{i,j})\in\mathbb{F}_q^{rt}\}.
\]

Notice that for any block $\{P_{i,1},\ldots,P_{i,r+1}\}$, we have
\[
	f_j(P_{i,1})=f_j(P_{i,2})=\cdots=f_j(P_{i,r+1})=f_j(Q_i).
\]
Let $f_{\bf a}\in V$ be the function with respect to ${\bf a}=(a_{i,j})\in\mathbb{F}_q^{rt}$, i.e.,
\[
 f_{\bf a}=\sum\limits_{i=0}^{r-1}\sum\limits_{j=1}^ta_{i,j}f_jz^i.
\]
Consider the code
\[
	C_{e,1} := \{(f_{\bf a}(P_{1,1}),\ldots,f_{\bf a}(P_{s,r+1}),a_{r-1,t}| f_{\bf a}\in V\}.
\]
Then we have:
\begin{theorem}\label{extend1}
Let $r$ be an integer and satisfy one of the following three conditions:
\begin{itemize}
\item[1.] $r=p^v-1$, or\\
\item[2.] $(r+1)|(q-1)$, or\\
\item[3.] $r+1=up^v$.
\end{itemize}
Assume that $2\le s\le \frac{q+1-2r}{r+1}$ is an integer. Then for any $1\le t\le s$, the code $C_{e,1}$ is an optimal LRC code with length $n=s(r+1)+1$, dimension $k=rt$ and locality $r$.
\end{theorem}
\begin{proof}
By the Hurwitz Genus Theorem, we have
\[
	-2=-2(r+1)+{\rm deg(Diff}(E/F))
\]
which means
\[
	2r= {\rm deg(Diff}(E/F)).
\]
It follows that there are at most $2r$ places of $F$ which are ramified in $E/F$. Therefore there are at least $q+1-2r\ge s(r+1)$ places which are unramified. For any place $P\in E/\mathbb{F}_q$ which is unramified, let $Q$ be the restriction of $P$ in $F/\mathbb{F}_q$, we can get $e(P|Q)=1$ and $f(P|Q)=1$. Hence we can find that the $s$ places $Q_1,\ldots,Q_s$ are splitting completely.

The code that generated by $C_{e,1}$ subtracting the last position is a subcode of an AG code with divisor $G$. Since ${\rm deg}(G)<n-1$ by $t\le s$, the dimension of the subcode equals the dimension of $V$. The last position of $C_{e,1}$ does not increase the dimension, thus the dimension of $C_{e,1}$ is $k$. Moreover, the $a_{r-1,t}$ is constant when ${\bf a}$ has been chosen, then the length of $C_{e,1}$ is exactly $s(r+1)+1$.

Suppose that the erased symbol is $f_{\bf a}(P_{u,v})$ where $P_{u,v}\in\{P_{u,1},\ldots,P_{u,r+1}\}$. For any $P_{u,h}$, we have
\begin{align*}
	f_{\bf a}(P_{u,h})&=\sum\limits_{i=0}^{r-1}\sum\limits_{j=1}^ta_{i,j}f_jz^i(P_{u,h})\\
	&=\sum\limits_{i=0}^{r-1}z^i(P_{u,h})\sum\limits_{j=1}^ta_{i,j}f_j(Q_u)\\
	&=\sum\limits_{i=0}^{r-1}c_iz^i(P_{u,h})
\end{align*}
where $c_i=\sum\limits_{j=1}^ta_{i,j}f_j(Q_u)$. If the $r$ values $f_{\bf a}(P_{u,h\neq v})$ are known, the values $c_i$ can be recovered by using the Lagrange interpolation. Since $z^i(P_{u,v})$ are known, the value $f_{\bf a}(P_{u,v})$ can be recovered. Similarly, we have $f_j(Q_1)=0$ for $1\le j\le t-1$ and the value of $f_t(Q_1)$ is known. The symbol
\[
\sum\limits_{j=1}^ta_{r-1,j}f_j(Q_1)=a_{r-1,t}f_t(Q_1)
\]
can be recovered by any $r$ values in the set $\{f_{\bf a}(P_{1,1}),\ldots,f_{\bf a}(P_{1,r+1})\}$.

It remains to prove that $C_{e,1}$ is optimal. If $a_{r-1,t}\neq 0$, then the zeros of $f_{\bf a}$ at most $(r+1)t-2$. It follows that the minimum distance $d$ at least $n-(r+1)t+2=n-rt-t+2$, which attains the Singleton bound. If $a_{r-1,t}=0$, then we have
\[
	f_{\bf a}=\sum\limits_{i=0}^{r-2}\sum\limits_{j=1}^ta_{i,j}f_jz^i+\sum\limits_{j=1}^{t-1}a_{r-1,j}f_jz^{r-1}.
\]
It follows that ${\rm deg}(f_{\bf a})_{\infty}\le (t-1)(r+1)+r-2$, which means that $f_{\bf a}$ has at most $t(r+1)-3$ zeros. This implies that the minimum distance $d$ is at least $n-t(r+1)+3-1=n-rt-t+2$. Therefore the code $C_{e,1}$ is optimal. The proof is completed.
\end{proof}
\begin{remark}
If we choose $f_1,\ldots,f_{j-1}$ in $\mathcal{L}((t-1)Q_{\infty}-Q_j)$ where $2\le j\le s$, we can also give an optimal LRC codes $C_{e,i}$. Moreover, we can extend $C_{e,j}$ to length $s(r+1)+r$ by adding $c_i$ in the proof, but it will not keep optimal.
\end{remark}

\begin{example}
Let $q=64$. Then for
\[
 r\in\{1,2, 3, 4, 5, 6, 7, 8, 9, 11, 12, 13, 15, 17, 20\},
\]
one can construct optimal LRCs with length $n=s(r+1)+1$ where $2\le s\le \frac{65-2r}{r+1}$, dimension $k=rt$ where $1\le t\le s$ and locality $r$. For instance, let $r=3$, then by Theorem \ref{extend1} we have optimal $[57,3t,59-4t]$ LRCs with locality 3 and $1\le t\le 14$. Let $r=4$, then we have optimal $[56,4t,58-5t]$ LRCs with locality 4 and $1\le t\le 11$. Let $r=5$, then we have optimal $[55,5t,57-6t]$ LRCs with locality 5 and $1\le t\le 9$. 

For $r\in\{1,3,7,15,31\}$, one can construct optimal $[s(r+1)+1,rt]$ LRCs with $2\le s\le \frac{64}{r+1}$ and $1\le t\le s$. Let $r=3$, then there exists an optimal $[65,3t,67-4t]$ LRC with locality 4 for any $1\le t\le 16$.
\end{example}

In \cite{Jin-optimal}, Jin, Ma and Xing have given an optimal LRC code with length $q+1$ on a rational field $E/\mathbb{F}_q$. Their construction is employed by the modified AG code and an automorphism group of $E/\mathbb{F}_q$. Since $E/\mathbb{F}_q$ is an rational function field, there exists a function $x$ such that $E/\mathbb{F}_q=\mathbb{F}_q(x)$. Let $f(x)$ be a quadratic polynomial $x^2+ax+b\in\mathbb{F}_q[x]$, and $\sigma$ be the automorphism defined by $\sigma(x)=\frac{1}{-bx-a}$. Then the order of $\sigma$ is $q+1$. Consider the subgroup of $\mbox{\textless}\sigma\mbox{\textgreater}$. Let $r$ be an integer such that $r+1|q+1$, and $\mathcal{G}$ be a subgroup of $\mbox{\textless}\sigma\mbox{\textgreater}$ with order $r+1$. From Proposition V.2 in \cite{Jin-optimal}, assume that $F/\mathbb{F}_q$ is the subfield of $E/\mathbb{F}_q$ such that $Gal(E/F)=\mathcal{G}$. Then we have
\begin{itemize}
	\item $[E:F]=r+1$
	\item There are exactly $\frac{q+1}{r+1}$ rational places of $F/\mathbb{F}_q$ which are splitting completely in $E/F$.
\end{itemize}

Based on the construction of Theorem \ref{extend1}, we can provide an optimal LRC code with length $q+2$ as follows.
\begin{theorem}\label{modifiedoptimal}
Let $r$ be an integer such that $r+1|q+1$, and $s$ be an integer with $2\le s\le\frac{q+1}{r+1}$. Then for any $1\le t\le s$, there exists an optimal LRC code with length $n=s(r+1)+1$, dimension $k=rt$ and locality $r$.
\end{theorem}
\begin{proof}
The case of $s<\frac{q+1}{r+1}$ is similar to Theorem \ref{extend1}, hence we only need to consider that $s=\frac{q+1}{r+1}$. Let $E/\mathbb{F}_q=\mathbb{F}_q(x)$ be a rational function, and $\mathcal{G}$ is a subgroup of $\mbox{\textless}\sigma\mbox{\textgreater}$ with order $r+1$. Therefore there are $s$ places $Q_1,\ldots,Q_s$ of $F/\mathbb{F}_q$ which are splitting completely. Assume that $P_{i,1},\ldots,P_{i,r+1}$ are places of $E/\mathbb{F}_q$ lying on $Q_i$. 

We can choose an element $x_1$ of $F/\mathbb{F}_q$ such that $(x_1)_{\infty}=Q_1$. It follows that $F/\mathbb{F}_q=\mathbb{F}_q(x_1)$. Moreover, we can choose $f_1,\ldots,f_t$ such that $\{f_j|j=1,\ldots,t\}$ forms a basis of $\in\mathcal{L}((t-1)Q_1)$ with 
\[
	0={\rm deg}(f_1)_{\infty}<\ldots<{\rm deg}(f_t)_{\infty}=t-1
\]
It is clear that $(x_1)_{\infty}=P_{1,1}+\cdots+P_{1,r+1}$, and $f_j\in\mathcal{L}((t-1)(P_{1,1}+\cdots+P_{1,r+1}))$ in $E/\mathbb{F}_q$. When $i\ge2$, we have $x_1$ and $f_j$ are constant functions on $\{P_{i,1},\ldots,P_{i,r+1}\}$.

Let $z$ denote an element in $E/\mathbb{F}_q$ such that $(z)_{\infty}=P_{2,1}$. We have $z\notin F/\mathbb{F}_q$ and consider the function space
\[
  V := \{\sum\limits_{i=0}^{r-1}\sum\limits_{j=1}^ta_{i,j}f_jz^i | (a_{i,j})\in\mathbb{F}_q^{rt}\}.
\]
Let 
\[
G=(t-1)(P_{1,1}+\cdots+P_{1,r+1})+(r-1)P_{2,1}.
\]
Then the set $V$ is contained in $\mathcal{L}(G)$ as a subspace with dimension $k=rt$. Define the function $\pi_1=\frac{1}{x_1}$, and $\pi_2=\frac{1}{z}$. The modified of $C_{e,1}$ is given as:
\begin{align*}
	C_{e,1}^m &:= \{(\pi_1^{t-1}f_{\bf a}(P_{1,1}),\ldots,\pi_1^{t-1}f_{\bf a}(P_{1,r+1}),\\
	&\pi_2^{r-1}f_{\bf a}(P_{2,1}),f_{\bf a}(P_{2,2}),\ldots,f_{\bf a}(P_{s,r+1}),a_{r-1,t}| f_{\bf a}\in V\}.
\end{align*}
It is clear that the code length is $n=s(r+1)+1$ and the code dimension is $k=rt$.

Suppose that the erased symbol is $f_{\bf a}(P_{u,v})$ where $P_{u,v}\in\{P_{u,1},\ldots,P_{u,r+1}\}$. If $u\neq 1,2$, the locality proof is similar to Theorem \ref{extend1}. For $u=1$, we consider the function
\begin{align*}
	f'(P_{1,h})&=\sum\limits_{i=0}^{r-1}\sum\limits_{j=1}^ta_{i,j}\pi_1^{t-1}f_jz^i(P_{1,h})\\
	&=\sum\limits_{i=0}^{r-1}z^i(P_{1,h})\sum\limits_{j=1}^ta_{i,j}\pi_1^{t-1}f_j(Q_1)\\
	&=\sum\limits_{i=0}^{r-1}c_iz^i(P_{1,h})
\end{align*}
where $c_i=a_{i,t}\pi_1^{t-1}f_t(Q_1)$. The value $\pi_1^{t-1}f_t(Q_1)$ is well defined since
\[
(\pi_1^{t-1}f_t)_{\infty}=(t-1)Q_1-(t-1)Q_1=0.
\]
If the $r$ values $f'(P_{u,h\neq v})$ are known, the values $c_i$ can be recovered by using the Lagrange interpolation. Since $z^i(P_{u,v})$ are known, therefore $f'(P_{u,v})$ can be recovered. Notice that $c_{r-1}=a_{r-1,t}\pi_1^{t-1}f_t(Q_1)$. Thus the symbol $a_{r-1,t}$ can also be recovered by any $r$ values in $f'(P_{1,h})$. For $u=2$, we consider the function
\begin{align*}
	f_{\bf a}(P_{2,h\neq1})&=\sum\limits_{i=0}^{r-1}\sum\limits_{j=1}^ta_{i,j}f_jz^i(P_{2,h\neq1})\\
	&=\sum\limits_{i=0}^{r-1}z^i(P_{2,h\neq1})\sum\limits_{j=1}^ta_{i,j}f_j(Q_2)\\
	&=\sum\limits_{i=0}^{r-1}d_iz^i(P_{2,h\neq1})
\end{align*} 
where $d_i=\sum\limits_{j=1}^ta_{i,j}f_j(Q_2)$, and the value
\begin{align*}
	\pi_2^{r-1}f_{\bf a}(P_{2,1})&=\sum\limits_{i=0}^{r-1}\sum\limits_{j=1}^ta_{i,j}\pi_2^{r-1}f_jz^i(P_{2,1})\\
	&=\sum\limits_{i=0}^{r-1}\pi_2^{r-1}z^i(P_{2,1})\sum\limits_{j=1}^ta_{i,j}f_j(Q_2)\\
	&=\sum\limits_{j=1}^ta_{r-1,j}f_j(Q_2)=d_{r-1}.
\end{align*} 
Therefore, for the case $v=1$, the symbol $d_{r-1}$ can be recovered by the $r$ values $f_{\bf a}(P_{2,h\neq1})$. For the case $v\neq1$, the values $d_i$ with $i\le r-2$ can be recovered by the $r-1$ values $f_{\bf a}(P_{2,h\neq1,v})-d_{r-1}z^{r-1}(P_{2,h\neq1,v})$, since $d_{r-1}$ is known in this case.

It remains to prove that $C_{e,1}^m$ is optimal. Suppose that $a_{r-1,t}\neq0$, the Hamming weight of the codeword is at least $n-(t-1)(r+1)-(r-1)=n-rt-t+2$. Then if $a_{r-1,t}=0$, we have $f_{\bf a}\in\mathcal{L}(G-P_{2,1})$. It follows that ${\rm deg}(f_{\bf a})_{\infty}\le (t-1)(r+1)+r-2=rt+t-3$, and the Hamming weight of the codeword is at least $n-(rt+t-3)-1=n-rt-t+2$. The proof is completed.
\end{proof}
\begin{example}
Let $q=64$. Then for $r=4,12$, one can construct optimal LRCs with length $n=s(r+1)+1$ where $2\le s\le \frac{65}{r+1}$, dimension $k=rt$ where $1\le t\le s$ and locality $r$ by Theorem \ref{extends}. For instance, let $r=4$, then we have optimal $[66,4t,68-5t]$ LRCs with locality 4 and $1\le t\le 13$.
\end{example}

\subsection{Construction of extend optimal $(r,3)$-LRCs}\label{section3.2}
In the construction of Theorem \ref{extend1}, we have give an extension of optimal LRCs. It is natural to consider that $s$-coordinates extension of optimal LRC codes. Thus, we have the following construction. Let $f_1,\ldots,f_t\in\mathcal{L}((t-1)Q_{\infty})$ are linearly independent over $F/\mathbb{F}_q$. Let $z\in E/\mathbb{F}_q$ with $(z)_{\infty}=P_{\infty}$. We also define the function space
\[
V := \{\sum\limits_{i=0}^{r-1}\sum\limits_{j=1}^ta_{i,j}f_jz^i | (a_{i,j})\in\mathbb{F}_q^{rt}\}.
\]

Notice that $f_j$ is constant on each sets $\{P_{i,1},\ldots,P_{i,r+1}\}$. Let $f_{\bf a}\in V$ be the function with respect to ${\bf a}=(a_{i,j})\in\mathbb{F}_q^{rt}$, i.e.,
\[
 f_{\bf a}=\sum\limits_{i=0}^{r-1}\sum\limits_{j=1}^ta_{i,j}f_jz^i.
\]
Consider the code
\begin{align*}
C_{e} := \{&(f_{\bf a}(P_{1,1}),\ldots,f_{\bf a}(P_{s,r+1}),\\
		&\sum\limits_{j=1}^ta_{r-1,j}f_j(Q_1),\ldots,\sum\limits_{j=1}^ta_{r-1,j}f_j(Q_s)| f_{\bf a}\in V\}.
\end{align*}
Then we have:
\begin{theorem}\label{extends}
Let $r$ be an integer and satisfy one of the following three conditions:
\begin{itemize}
\item[1.] $r=p^v-1$, or\\
\item[2.] $(r+1)|(q-1)$, or\\
\item[3.] $r+1=up^v$.
\end{itemize}
Assume that $2\le s\le \frac{q+1-2r}{r+1}$ is an integer. Then for any $1\le t\le s$, the code $C_{e}$ is an $(r,3)$-LRC with length $n=s(r+2)$, dimension $k=rt$ and locality $r$. Moreover, if $t=s$, then $C_{e}$ is optimal.
\end{theorem}
\begin{proof}
Note that $\sum\limits_{j=1}^ta_{r-1,j}f_j(Q_i)$ are constant for $i=1,\ldots,s$ when $Q_1,\ldots,Q_s$ and $\textbf{a}$ have been chosen. It follows from Theorem \ref{extend1} that the parameters of $C_{e}$ are $n=s(r+2)$ and $k=rt$. 

Consider the block $\{P_{u,1},\ldots,P_{u,r+1},\sum\limits_{j=1}^ta_{r-1,j}f_j(Q_u)\}$. For any $P_{u,h}$, we have
\begin{align*}
	f_{\bf a}(P_{u,h})&=\sum\limits_{i=0}^{r-1}\sum\limits_{j=1}^ta_{i,j}f_jz^i(P_{u,h})\\
	&=\sum\limits_{i=0}^{r-1}z^i(P_{u,h})\sum\limits_{j=1}^ta_{i,j}f_j(Q_u)\\
	&=\sum\limits_{i=0}^{r-1}c_iz^i(P_{u,h})
\end{align*}
where $c_i=\sum\limits_{j=1}^ta_{i,j}f_j(Q_u)$. If the $r$ values $f_{\bf a}(P_{u,v\neq h})$ are known, the values $c_i$ can be recovered by using the Lagrange interpolation. Therefore $f_{\bf a}(P_{u,v})$ and 
\[
c_{r-1}=\sum\limits_{j=1}^ta_{r-1,j}f_j(Q_u)
\]
can be recovered since $z(P_{u,h})$ is known. If the $r-1$ values $f_{\bf a}(P_{u,v_i})$ for $i=1,\ldots,r-1$ and $\sum\limits_{j=1}^ta_{r-1,j}f_j(Q_u)$ are known, then we have
\[
	f_{\bf a}(P_{u,v_i})=\sum\limits_{i=0}^{r-2}c_iz^i(P_{u,v_i})+\sum\limits_{j=1}^ta_{r-1,j}f_j(Q_u)z^{t-1}(P_{u,v_i}),i=1,\ldots,r-1.
\]
Let $f^u_{\bf a}=\sum\limits_{i=0}^{r-2}c_iz^i$, then $f^u_{\textbf{a}}$ can be recovered by using the Lagrange interpolation from the $r-1$ equations. Therefore the two values $f_{\bf a}(P_{u,v\neq v_i})$ can be recovered. Then the proof for that $C_e$ has $(r,3)$-locality is completed.

Consider $t=s$. Note that $f_j\in\mathcal{L}((t-1)Q_{\infty})$ for $j=1,\ldots,t$. Suppose that there are at most $t-1$ zeros in
\[
\sum\limits_{j=1}^ta_{r-1,j}f_j(Q_1),\ldots,\sum\limits_{j=1}^ta_{r-1,j}f_j(Q_s).
\]
Since ${\rm deg}(f_{\bf a})_{\infty}\le (t-1)(r+1)+r-1$, the Hamming weight of any codeword in $C_e$ is at least
\begin{align*}
&n-(t-1)(r+1)-r+1-t+1\\
&=n-tr-2t+3\\
&=n-k-(\frac{k}{r}-1)(3-1)+1.
\end{align*}
Therefore the minimum distance of $C_e$ attains the Singleton upper bound of $(r,3)$-LRCs. If there are at least $t$ zeros in
\[
\sum\limits_{j=1}^ta_{r-1,j}f_j(Q_1),\ldots,\sum\limits_{j=1}^ta_{r-1,j}f_j(Q_s),
\]
then we can obtain that these $s$ values are all equal to 0 since $f_1,\ldots,f_t$ are linearly independent over $F/\mathbb{F}_q$. Thus we have $a_{r-1,j}=0$ for $j=1,\ldots,s$. It follows that
\[
{\rm deg}(f_{\bf a})_{\infty}\le (t-1)(r+1)+r-2,
\]
which means $f_{\bf a}$ has at most $t(r+1)-3$ zeros. Then the Hamming weight of any codeword in $C_e$ is at least $n-t(r+1)+3-s$. Therefore $C_e$ is optimal when $t=s$.
\end{proof}

Similar to the definition of $C_{e,1}^m$, we can define the following code:
\begin{align*}
	C_{e}^m := \{ &(\pi_1^{t-1}f_{\bf a}(P_{1,1}),\ldots,\pi_1^{t-1}f_{\bf a}(P_{1,r+1}),\\
	&\pi_2^{r-1}f_{\bf a}(P_{2,1}),f_{\bf a}(P_{2,2}),\ldots,f_{\bf a}(P_{s,r+1}),\\
	&a_{r-1,t},\sum\limits_{j=1}^ta_{r-2,j}f_j(Q_2),\ldots,\sum\limits_{j=1}^ta_{r-1,j}f_j(Q_s)| f_{\bf a}\in V\}.
\end{align*}
From the construction of $C_{e,1}^m$ and $C_{e}$, we can get the following corollary.
\begin{corollary}
Suppose that $r+1|q+1$. The code $C_{e}^m$ is an $(r,3)$-LRC code with $n=(s+1)(r+1)$ and $k=rt$. If $t=s$, then $C_{e}^m$ is $d$-optimal.
\end{corollary}
\begin{proof}
The proof of the parameters is similar to that of Theorem \ref{extends}. For the locality, we only need to consider the first two blocks, and it was shown in Theorem \ref{modifiedoptimal}.
\end{proof}
\begin{example}
Let $q=64$. Then by Theorem \ref{extends}, let $r=3$, one can construct an optimal $[68,3t,71-5t]$ LRC with locality $(3,3)$ for any $1\le t\le 16$. Let $r=4$, one can construct an optimal $[70,4t,73-6t]$ LRC with locality $(4,3)$ for any $1\le t\le 13$.
\end{example}

\section{Extension from Roth-Lempel type LRCs on function fields}
A well-known class of non-RS MDS code is the Roth-Lempel code which introduced by Roth and Lempel \cite{Roth-nRS}. We will consider the Roth-Lempel type extension of the optimal LRCs from function fields in this section.

With the same notations as Section \ref{section3.1}, we consider the code:
\[
C_{RL,1} := \{(f_{\bf a}(P_{1,1}),\ldots,f_{\bf a}(P_{s,r+1}),a_{r-2,t},a_{r-1,t},| f_{\bf a}\in V\}.
\]
Then we have the following result.
\begin{corollary}\label{extend1rl}
Let $r$ be an integer and satisfy one of the following three conditions:
\begin{itemize}
\item[1.] $r=p^v-1$, or\\
\item[2.] $(r+1)|(q-1)$, or\\
\item[3.] $r+1=up^v$.
\end{itemize}
Assume that $2\le s\le \frac{q+1-2r}{r+1}$ is an integer. Then for any $1\le t\le s$, the code $C_{RL,1}$ is an LRC code with length $n=s(r+1)+1$, dimension $k=rt$ and locality $r$. Moreover, the minimum distance of $C_{RL,1}$ satisfies $d\ge n-rt-t+1$.
\end{corollary}
\begin{proof}
The parameters of $C_{R,1}$ and the locality of blocks $\{P_{u,1},\ldots,P_{u,r+1}\}$ with $2\le u\le s$ are similar to Theorem \ref{extend1}. We only need to consider the locality of block $\{f_{\bf a}(P_{1,1}),\ldots,f_{\bf a}(P_{1,r+1}),a_{r-2,t},a_{r-1,t}\}$. Note that for any $1\le h\le r+1$, we have
\begin{align*}
	f_{\bf a}(P_{1,h})&=\sum\limits_{i=0}^{r-1}\sum\limits_{j=1}^ta_{i,j}f_jz^i(P_{1,h})\\
	&=\sum\limits_{i=0}^{r-1}z^i(P_{1,h})\sum\limits_{j=1}^ta_{i,j}f_j(Q_1)\\
	&=\sum\limits_{i=0}^{r-1}c_iz^i(P_{1,h})
\end{align*}
where $c_i=a_{i,t}f_j(Q_1)$. Since $f_j(Q_1)$ is known, the symbols of $a_{i,t}$ can be recovered by any $r$ values of $f_{\bf a}(P_{1,h})$.

Now we consider the minimum distance of $C_{RL,1}$ and we distinguish into four cases:
\begin{itemize}
\item[1.] If $a_{r-2,t}\neq0$ and $a_{r-1,t}\neq0$, then the vector $(f_{\bf a}(P_{1,1}),\ldots,f_{\bf a}(P_{s,r+1}))$ has Hamming weight at least $n-rt-t+2$ from Section \ref{section2.4}. Therefore in this case the codeword also has Hamming weight at least $n-rt-t+2$.
\item[2.] If $a_{r-2,t}\neq0$ and $a_{r-1,t}=0$, then the vector $(f_{\bf a}(P_{1,1}),\ldots,f_{\bf a}(P_{s,r+1}))$ has Hamming weight at least $n-rt-t+3$ from Section \ref{section3.1}. Therefore in this case the codeword has Hamming weight at least $n-rt-t+2$.
\item[3.] If $a_{r-2,t}=0$ and $a_{r-1,t}=0$, then we have ${\rm deg}(f_{\bf a})_{\infty}\le t(r+1)-4$. Thus the vector $(f_{\bf a}(P_{1,1}),\ldots,f_{\bf a}(P_{s,r+1}))$ has Hamming weight at least $n-rt-t+4$, and the codeword has Hamming weight at least $n-rt-t+2$.
\item[4.] Suppose that $a_{r-2,t}=0$ and $a_{r-1,t}\neq0$. We consider the zeros in the block
\[
\{f_{\bf a}(P_{1,1}),\ldots,f_{\bf a}(P_{1,r+1}),0,a_{r-1,t}\}.
\]
Note that any rational place $P_{1,h}\in E/\mathbb{F}_q$ is corresponding to an element $\alpha_{1,h}\in\mathbb{F}_q$. If there dose not exist $r-1$ elements $\alpha_{1,i_1},\ldots,\alpha_{1,i_{r-1}}$ such that $\alpha_{1,i_1}+\cdots+\alpha_{1,i_{r-1}}=0$, then the number of the zeros in the above block is equal to that of the zeros in Case 1 and the codeword has Hamming weight at least $n-rt-t+2$.

Otherwise, the vector $(f_{\bf a}(P_{1,1}),\ldots,f_{\bf a}(P_{s,r+1}))$ has Hamming weight at least $n-rt-t+2$ from Section \ref{section2.4}. Therefore in this case the codeword has Hamming weight at least $n-rt-t+1$.
\end{itemize}
According to above discussion, we have $d\ge n-rt-t+1$. The proof is completed.
\end{proof}

In the proof, we can obtain that the Roth-Lempel type extension of an optimal LRC from rational function fields is also be optimal if the elements satisfying some conditions. The optimal conditions is similar to the MDS conditions of classic Roth-Lempel codes \cite{Roth-nRS}.

Combining the above result with Theorem \ref{modifiedoptimal}, we can obtain the following corollary:
\begin{corollary}\label{modifiedextend1rl}
Let $r$ be an integer such that $r+1|q+1$, and $s$ be an integer with $2\le s\le\frac{q+1}{r+1}$. Then for any $1\le t\le s$, there exists an LRC code with length $n=s(r+1)+1$, dimension $k=rt$, locality $r$ and minimum distance $d\ge n-rt-t+1$.
\end{corollary}

\begin{remark}
With the same notations as Section \ref{section3.2}, we can also consider the Roth-Lempel type extension code as follows:
\begin{align*}
C_{RL} := \{&(f_{\bf a}(P_{1,1}),\ldots,f_{\bf a}(P_{s,r+1}),\\
		&\sum\limits_{j=1}^ta_{r-2,j}f_j(Q_1),\sum\limits_{j=1}^ta_{r-1,j}f_j(Q_1),\\
		&\ldots,\sum\limits_{j=1}^ta_{r-2,j}f_j(Q_s),\sum\limits_{j=1}^ta_{r-1,j}f_j(Q_s)| f_{\bf a}\in V\}.
\end{align*}
One can prove that $C_{RL}$ is a $(r,3)$-LRC. By a calculation, we can obtain
\[
d(C_{RL})\le n-k-(\frac{k}{r}-1)(3-1)+1-s.
\]
Then $C_{RL}$ is not an optimal $(r,3)$-LRC.
\end{remark}

\section{Extension of optimal LRCs from elliptic function fields}
In this section, we will extend the optimal LRCs given by Li et al. in \cite{Li-optimal} and Ma and Xing in \cite{Ma-optimal}. Let $E/\mathbb{F}_q$ be an elliptic function field and $\mathcal{G}\subset Aut(E/\mathbb{F}_q)$ is an automorphism subgroup with order $r+1$. Suppose that $F/\mathbb{F}_q=E^{\mathcal{G}}$ be the subfield of $E/\mathbb{F}_q$. Then by Lemma \ref{funcfield} we have $F/\mathbb{F}_q$ is a rational function field, and $E/F$ is a finite separable extension with $[E:F]=r+1$. Let $Q_1,\ldots,Q_s,Q_{\infty}$ and $P_{1,1},\ldots,P_{s,r+1},P_{\infty}$ be defined as Section \ref{section2.5}.
\subsection{Extension of optimal elliptic LRCs}
Since $F/\mathbb{F}_q$ is a rational function field, the function space $\mathcal{L}(Q_1-Q_{\infty})$ has dimension $\ell(Q_1-Q_{\infty})=1$. It means that we can find a non-zero $\hat{z}$ such that $(\hat{z})=Q_1-Q_{\infty}$. Consider the function space
\[
V_{E,1} := \{\sum\limits_{j=1}^{t-1}a_{0,j}z_j+a_{0,t}z^{t-1}+\sum\limits_{i=1}^{r-1}(\sum\limits_{j=1}^{t-1}a_{i,j}z^{j-1})\omega_i | a_{i,j}\in\mathbb{F}_q\}.
\]
where $z_j=\hat{z}z^{j-1}$ for $j=1,\ldots,t-1$. From the linearly independence of $\omega_i$ for $i=0,\ldots,r-1$, one can check that the dimension of $V_{E,1}$ is $rt-r+1$. Let $f_{\bf a}\in V_{E,1}$ be the function with respect to ${\bf a}=(a_{i,j})\in\mathbb{F}_q^{rt-r+1}$, i.e.,
\[
 f_{\bf a}=\sum\limits_{j=1}^{t-1}a_{0,j}z_j+a_{0,t}z^{t-1}+\sum\limits_{i=1}^{r-1}(\sum\limits_{j=1}^{t-1}a_{i,j}z^{j-1})\omega_i.
\]
Consider the code
\[
	C_{E,1} := \{(f_{\bf a}(P_{1,1}),\ldots,f_{\bf a}(P_{s,r+1}),a_{0,t})| f_{\bf a}\in V_{E,1}\}
\]
and then we present the following results.
\begin{theorem}\label{eclrc1}
Suppose that $\#E(\mathbb{F}_q)=N$. Let $2\le s\le \lfloor\frac{N-2r-4}{r+1}\rfloor$ be an integer. Then for any $1\le t\le s$, the code $C_{E,1}$ is an optimal LRC code with length $n=s(r+1)+1$, dimension $k=rt-r+1$ and locality $r$.
\end{theorem}
\begin{proof}
By the Hurwitz Genus Theorem, we have
\[
	2g(E)-2=0=(r+1)(0-2)+{\rm deg(Diff}(E/F)).
\]
The different exponent of $Q_{\infty}$ is at least $r$, therefore there are at most other $r+2$ rational places which are ramified in $E/F$. Without the $r+1$ points $P_{0,1},\ldots,P_{0,j+1}$, we can find that at most $N-2r-4\ge s(r+1)$ places of $E/\mathbb{F}_q$ are splitting completely in $E/F$. Then the length and dimension of the algebraic geometry code $C_{E,1}$ are indeed $s(r+1)+1$ and $rt-r+1$ respectively.

The locality of the blocks $\{P_{i,1},\ldots,P_{i,r+1}\}$ for $i=2,\ldots,s$ have shown in \cite{Li-optimal}, and we only need to prove the locality for $\{P_{1,1},\ldots,P_{1,r+1},a_{0,t}\}$. Since $\hat{z}(Q_1)=0$, for any $P_{1,v}$, we have
\begin{align*}
	f_{\bf a}(P_{1,v})&=a_{0,t}z^{t-1}(Q_1)+\sum\limits_{i=1}^{r-1}(\sum\limits_{j=1}^{t-1}a_{i,j}z^{j-1}(Q_1))\omega_i(P_{1,v})\\
	&=a_{0,t}z^{t-1}(Q_1)+\sum\limits_{i=1}^{r-1}\omega_i(P_{1,v})\sum\limits_{j=1}^{t-1}a_{i,j}z^{j-1}(Q_1)\\
	&=\sum\limits_{i=0}^{r-1}c_i\omega_i(P_{1,v})
\end{align*}
where $c_0=a_{0,t}z^{t-1}(Q_1)$ and 
\[
c_i=\sum\limits_{j=1}^{t-1}a_{i,j}z^{j-1}(Q_1)\ \mbox{for}\ 1\le i\le r-1.
\]
Since the values of $\omega_i(P_{1,v})$ are known, we can calculate the values of $c_i$ from any $r$ subset of the $f_{\bf a}(P_{1,v})$ from Section \ref{section2.5}. Therefore the value of any $f_{\bf a}(P_{1,h})$ can be recovered from other $r$ values of $f_{\bf a}(P_{1,v})$. Note that $a_{0,t}z^{t-1}(Q_1)=c_0$ where $z^{t-1}(Q_1)$ is known and $c_0$ can be computed by any $r$ values of $f_{\bf a}(P_{1,v})$. Thus $a_{0,t}$ can be recovered.

It remains to prove that $C_{E,1}$ is optimal. If $a_{0,t}\neq0$, then there exists a divisor $G=(t-1)(P_1+\cdots P_{r+1})$ such that $V_{E,1}\subset \mathcal{L}(G)$. By the design distance and the Singleton Bound, we have the minimum distance $d$ of $C_{E,1}$ satisfying
\[
 n-(t-1)(r+1)\le d\le n-rt+r-1-t+2=n-(t-1)(r+1).
\]
Suppose that $a_{0,t}=0$. Since $\omega_i\in\mathcal{L}(P_1+\cdots+P_{i+1})$, it follows that
\[
	{\rm deg}((f_{\bf a})_{\infty})\le {\rm deg}((z^{t-2}\omega_{r-1})_{\infty})=(t-1)(r+1)-1.
\]
Because of $f_{\bf a}$ is a rational function, then there are at most $(t-1)(r+1)$ zeros for any code in $C_{E,1}$. It means that the Hamming weight of any code in $C_{E,1}$ is again at least $n-(t-1)(r+1)$, which attains the Singleton Bound.
\end{proof}
\begin{example}
Let $q=64$. Then there exists an elliptic curve $E$ such that $\#Aut(E/\mathbb{F}_{64})=24$. For any $r\in\{2,3,5,7,11,23\}$, one can construct optimal LRCs code with length $s(r+1)$ and locality $r$ where
\begin{itemize}
\item If $r=2$, then $2\le s\le 26$ and $k=2t+1$ for $1\le t\le s$.
\item If $r$ odd, then $2\le s\le \lfloor\frac{78-r}{r+1}\rfloor$ and $k=rt-r+1$ for $1\le t\le s$.
\end{itemize}
For instance, we can get an optimal $[79,2t+1,79-3t]$ LRC with locality 2 for any $1\le t\le 25$, and we also can get an optimal $[73,3t-2,77-4t]$ LRC with locality 3 for any $1\le t\le 18$.
\end{example}
\subsection{Construction of optimal $(r,3)$-LRCs from elliptic curves}
Similar to the case of rational function field in Theorem \ref{extends}, we can also construct $(r,3)$-LRCs from elliptic AG codes under certain conditions. Consider the function space $V_E$, and let $f_{\bf a}\in V_E$ be a function with respect to ${\bf a}=(a_{i,j})\in\mathbb{F}_q^{rt-r+1}$. For our purpose, we need the following lemma.
\begin{lemma}\label{eclrcr}
For any points set $\{P_{u,1},\ldots,P_{u,r+1}\}$ with $u\in (1,\ldots,s)$, let $M_u$ be defined as Lemma \ref{eclrc} and $M'_u$ be a corresponding matrix with
\begin{equation}\nonumber
M'_u=\begin{pmatrix}
&\omega_1(P_{u,1})&\omega_2(P_{u,1})&\ldots&\omega_{r-1}(P_{u,1})\\
&\omega_1(P_{u,2})&\omega_2(P_{u,2})&\ldots&\omega_{r-1}(P_{u,2})\\
&\vdots &\vdots &\ddots &\vdots\\
&\omega_1(P_{u,r+1})&\omega_2(P_{u,r+1})&\ldots&\omega_{r-1}(P_{u,r+1})
\end{pmatrix}
\end{equation}
Suppose that every $r\times r$ submatrix of $M_u$ is invertible and every $r-1\times r-1$ submatrix of $M'_u$ is invertible. Let 
\[
\{f_{\bf a}(P_{u,1}),\ldots,f_{\bf a}(P_{u,r+1}),\sum\limits_{j=1}^ta_{0,j}z^{j-1}(Q_1)\}
\]
be the block with respect to ${\bf a}$. Then any two values in this block can be calculated from the other $r$ values in this block.
\end{lemma}
\begin{proof}
Since the function $z$ has same values on $P_{u,v}$ for $v=1,\ldots,r+1$, we have
\[
	f_{\bf a}(P_{u,v})=\sum\limits_{i=0}^{r-1}c_i\omega_i(P_{u,v})
\]
where $c_i=\sum\limits_{j=1}^{t-1}a_{i,j}z^{j-1}(Q_u)$ with $0\le i\le r-1$ for any $P_{u,v}$. When $c_0$ is unknown, since every $r\times r$ submatrix of $M_u$ is invertible, we can calculate $c_i$ from the values of any subset of $r$ of the $f_{\bf a}(P_{u,v})$. Therefore $c_0$ and the other values in $\{f_{\bf a}(P_{u,v})\}$ can be recovered. Suppose that the values of $c_0$ and $r-1$ values $f_{\bf a}(P_{u,v_j})$ for $j=1,\ldots,r-1$ are known. Then we have
\begin{equation}\nonumber
\begin{pmatrix}
&\omega_1(P_{u,v_1})&\ldots&\omega_{r-1}(P_{u,v_1})\\
&\omega_1(P_{u,v_2})&\ldots&\omega_{r-1}(P_{u,v_2})\\
&\vdots &\ddots &\vdots\\
&\omega_1(P_{u,v_{r-1}})&\ldots&\omega_{r-1}(P_{u,v_{r-1}})
\end{pmatrix}
\begin{pmatrix}
c_1\\
c_2\\
\vdots\\
c_{r-1}
\end{pmatrix}
=
\begin{pmatrix}
f_{\bf a}(P_{u,v_1})-c_0\\
f_{\bf a}(P_{u,v_2})-c_0\\
\vdots\\
f_{\bf a}(P_{u,v_{r-1}})-c_0
\end{pmatrix}.
\end{equation}
Since every $r-1\times r-1$ submatrix of $M'_u$ is invertible, we can also calculate the values of the $c_i$ from the values of any
subset of $r-1$ of the $f_{\bf a}(P_{u,v})$. Thus the other two values of $f_{\bf a}(P_{u,v})$ can be recovered.
\end{proof}
Let us consider the following code:
\begin{align*}
C_{E} := \{&(f_{\bf a}(P_{1,1}),\ldots,f_{\bf a}(P_{s,r+1}),\\
		&\sum\limits_{j=1}^ta_{0,j}z^{j-1}(Q_1),\ldots,\sum\limits_{j=1}^ta_{0,j}z^{j-1}(Q_s)| f_{\bf a}\in V\}.
\end{align*}
Then we have the following result.
\begin{theorem}\label{eclrcall}
Suppose that $\#E(\mathbb{F}_q)=N$. Let $2\le s\le \lfloor\frac{N-2r-4}{r+1}\rfloor$ be an integer. If for any $1\le u\le s$, every $r-1\times r-1$ submatrix of $M'_u$ is invertible. Then for any $1\le t\le s$, the code $C_{E}$ is an $(r,3)$-LRC code with length $n=s(r+2)$, dimension $k=rt-r+1$ and locality $r$. Moreover, if $t=s$, then $C_{E}$ is optimal.
\end{theorem}
\begin{proof}
From Theorem \ref{eclrc1} the length and dimension of the code $C_{E}$ are indeed $s(r+1)+1$ and $rt-r+1$ respectively. By Lemmas \ref{eclrc} and \ref{eclrcr}, the code $C_{E}$ has locality $r$. It remains to prove that the $C_{E}$ is optimal when $t=s$.

Note that $\{z^j|j=0,\ldots,t-1\}$ are linearly independent. The number of the zeros in $\{\sum\limits_{j=1}^ta_{0,j}z^{j-1}(Q_u)|u=1,\ldots,s\}$ is either at most $t-1$ or $s$. For the first case, we have
\[
{\rm deg}((f_{\bf a})_{\infty})\le (t-1)(r+1).
\]
Thus the Hamming weight of any code in $C_E$ is at least $n-(t-1)(r+2)$. By the Singleton bound of $(r,3)$-LRCs, the minimum distance $d$ of $C_{E}$ satisfies
\[
 n-(t-1)(r+2)\le d\le n-rt+r-1-2(t-1)+1=n-(t-1)(r+2).
\]
Suppose that all $s$ values are 0. Since $\omega_i\in\mathcal{L}(P_1+\cdots+P_{i+1})$, it follows that $a_{0,j}=0$ with $j=1,\ldots,t$ and 
\[
{\rm deg}((f_{\bf a})_{\infty})\le (t-1)(r+1)-1.
\]
Then the Hamming weight of any code in $C_E$ is at least $n-(t-1)(r+1)-s+1$ and
\[
n-(t-1)(r+1)-s+1\le d\le n-(t-1)(r+2).
\]
Therefore if $t=s$, then we have $C_E$ must be optimal.
\end{proof}

\section{Conclusion}
In this paper, we gave some new constructions of optimal LRCs and optimal $(r,3)$-LRCs by extending RS-type blocks. For rational function fields, we extended the LRCs given in \cite{Tamo-family} and \cite{Jin-optimal}. Then we showed that these codes hold optimal if we only extend one RS-type block. It means that one can construct optimal LRCs with length $q+2$ and locality $r$ where $(r+1)|(q+1)$. In addition, we extended all of the RS-type blocks of an LRC, and we obtained an $(r,3)$-optimal LRC. We also showed that the LRCs extended by the Roth-Lempel type are optimal under some special conditions. For elliptic LRCs, we constructed optimal LRCs by extending one block directly. 

The constructions given in this paper can be also generalized to more other LRCs from algebraic curves in~\cite{Barg-lrcalg}. However, for ensuring the locality of these extended codes, more explicit conditions may be required. An interesting direction for the future research is to  consider the optimal extension of LRCs from more algebraic curves.  

\subsubsection*{\bf Data Availability}
No data.

\subsubsection*{\bf Declarations}
Conflict of interest The authors have no Conflict of interest to declare that are relevant to the content of this
article.

\subsubsection*{\bf Acknowledgement}
This work is supported by the National Natural Science Foundation of China (No. 12441107) and	Guangdong Major Project of Basic and Applied Basic Research (No. 2019B030302008),  and 
  Guangdong Provincial Key Laboratory of Information Security Technology (No. 2023B1212060026). 

\bibliography{lrc-references}
	
\end{document}